\documentclass[sigconf]{aamas} 

\usepackage{amsmath,amsfonts}
\usepackage{subfigure}
\usepackage{bm}
\usepackage{mathrsfs}
\usepackage{framed}
\usepackage{cuted}
\usepackage{xcolor}
\usepackage{balance} 

\newtheorem{remark}{Remark}[section]

\newcommand{\lm}{\bm{\theta}}
\newcommand{\bx}{\bm{x}}
\newcommand{\E}{\mathbb{E}}
\newcommand{\prob}{\mathbb{P}}
\newcommand{\loss}{\mathcal{L}}
\newcommand{\NashStableSet}{\mathscr{N}_{\text{stable}}}

\setcopyright{ifaamas}
\acmConference[AAMAS '21]{Proc.\@ of the 20th International Conference on Autonomous Agents and Multiagent Systems (AAMAS 2021)}{May 3--7, 2021}{London, UK}{U.~Endriss, A.~Now\'{e}, F.~Dignum, A.~Lomuscio (eds.)}
\copyrightyear{2021}
\acmYear{2021}
\acmDOI{}
\acmPrice{}
\acmISBN{}

\acmSubmissionID{???}

\title[AAMAS-2021 Formatting Instructions]{Incentive Mechanism Design for Federated Learning: Hedonic Game Approach}

\author{Cengis Hasan}
\affiliation{
  \institution{University of Luxembourg}
  \department{SnT - Interdisciplinary Centre for Security, Reliability and Trust}
  }
\email{cengis.hasan@uni.lu}

\begin{abstract}
Incentive mechanism design is crucial for enabling federated learning. We deal with clustering problem of agents contributing to federated learning setting. Assuming agents behave selfishly, we model their interaction as a stable coalition partition problem using hedonic games where agents and clusters are the players and coalitions, respectively. We address the following question: is there a family of hedonic games ensuring a Nash-stable coalition partition? We propose the Nash-stable set which determines the family of hedonic games possessing at least one Nash-stable partition, and analyze the conditions of non-emptiness of the Nash-stable set. Besides, we deal with the decentralized clustering. We formulate the problem as a non-cooperative game and prove the existence of a potential game.
\end{abstract}


\newcommand{\BibTeX}{\rm B\kern-.05em{\sc i\kern-.025em b}\kern-.08em\TeX}

\begin{document}


\pagestyle{fancy}
\fancyhead{}


\maketitle 


\section{Introduction}
Data protection is a major concern. If we do not trust someone withholding our data, we may opt for federated learning by privately developing intelligent systems to create privacy-preserving AI. \textit{Federated learning} enables privacy-preserving machine learning in a decentralized way \cite{bib:Li}. It is used in situations where data is distributed among different agents and training is impossible due to the difficulty to collect data centrally. All data is kept on device while a shared (global) learning model is trained in each device and aggregated (combined) centrally. Formally, we consider the following setting: i) data owner \textit{agents} which locally trains the shared learning model, and ii) \textit{model aggregating entity} (MAE) which combines learning model of its own with the agents. MAE and agents contribute to the same shared learning model. Federated learning has been identified as a distributed machine learning framework which sees rapid advances and broad adoption in next generation networking and edge systems \cite{bib:MehdiBennis,bib:Li,bib:OpenProblems,bib:CongShen,bib:WalidSaad1,bib:WalidSaad2,bib:MungChiang,bib:YuanmingShi}. Obviously, the motivation to implement federated learning is to reduce the variance in a learned model by accessing more data.

A very crucial question is how would MAE motivate the agents to participate in federated learning. Designing the mechanism of agents' incentives can be performed by utilizing various frameworks such game theory, auction theory, etc \cite{bib:WalidSaad1}. Any \textit{clustering} among agents (players) being able to make strategic decisions becomes a \textit{coalition formation game} when the players --for various individual reasons-- may wish to belong to a relative \textit{small coalition} rather than the \textit{grand coalition}--the set of all players. Players' moves from one to another coalition are governed by a set of rules. Basically, an agent (player) will move to a new coalition when it may obtain a better gain from this coalition. We shall not consider any permission requirements, which means that a player is always accepted by a coalition to which the player is willing to join. Based on those rules, the crucial question in the game context is \textit{how a stable partition exists}. \textit{This is essential to enable federated learning}. 

We study the \textit{hedonic} coalition formation game model of the agents and analyze the \textit{Nash stability} \cite{bib:hajdukova}. A coalition formation game is called \textit{hedonic} if each player's preferences over partitions of players depend only on the members of his/her coalition. Finding a stable coalition partition is the main question in a coalition formation game. We refer to \cite{bib:Existenceofstabilityinhedoniccoalitionformationgames} discussing the stability concepts associated to \textit{hedonic conditions}. In the sequel, we concentrate on the Nash stability. The definition of the Nash stability is quite simple: \textit{a partition of players is Nash stable whenever no player deviates from its coalition to another coalition in the partition}.

In this work, we deal with the following problem: having coalitions associated with their gain, we seek the answer of how must the coalition gain be allocated to the players in order to obtain a stable coalition partition. Clearly, the fundamental question is to determine which gain allocation methods may ensure a Nash-stable partition. Note that the answer of this enables to find the family of hedonic games that possess at least one Nash-stable partition of players. We first propose the definition of \textit{the Nash-stable set} which is the set of all possible allocation methods resulting in Nash-stable partitions. We show that additively separable and symmetric gain allocation always ensures Nash-stable partitions. Moreover, our work aims also at finding the partitions in a decentralized setting which basically corresponds to finding stable decentralized clustering. We model this problem as a non-cooperative game and show that such a game is a potential game. 

A recent work that considers the clustering of agents in the form of hedonic games can be found in \cite{bib:Donahue} where the authors study the agents decisions to participate in federated learning setting in case of a biased global model.
In \cite{bib:HierarchicalIncentiveMechanismDesign} a federated learning based privacy-preserving approach is proposed to facilitate collaborative machine learning among multiple model owners in mobile crowdsensing. Another work in \cite{bib:IncentiveDesignandDifferentialPrivacyBasedFederatedLearning} implements mechanism design and differential privacy where an objectives-first approach is considered for designing incentives toward desired objectives; the differential privacy can provide a theoretical guarantee for users' privacy in federated learning participation. In \cite{bib:IncentiveMechanismforFederatedLearninginWirelessCellularnetwork}, an incentive mechanism between a base station and mobile users as an auction game is formulated where the base station is an auctioneer and the mobile users are the sellers. In \cite{bib:IncentiveMechanismDesignforFederatedLearningwithMulti-DimensionalPrivateInformation}, the authors consider a multidimensional contract-theoretic approach on optimal incentive mechanism design, in the presence of users' multi-dimensional private information including training cost and communication delay. The work in \cite{bib:AnIncentiveMechanismDesignforEfficientEdgeLearningbyDeepReinforcementLearningApproach} deals with a deep reinforcement learning based approach to design the incentive mechanism and find the optimal trade-off between model training time and parameter server's payment. The authors in \cite{bib:MotivatingWorkersinFederatedLearning:AStackelbergGamePerspective} analyze the influence of heterogeneous clients on federated learning convergence, and propose an incentive mechanism to balance the time delay of each iteration. For a recent survey on mechanism design for federated learning, we refer to the paper in \cite{bib:ASurveyofIncentiveMechanismDesignforFederatedLearning}.

\section{Motivation and Problem Description}
\begin{figure*}
    \centering
    \includegraphics[width=\textwidth]{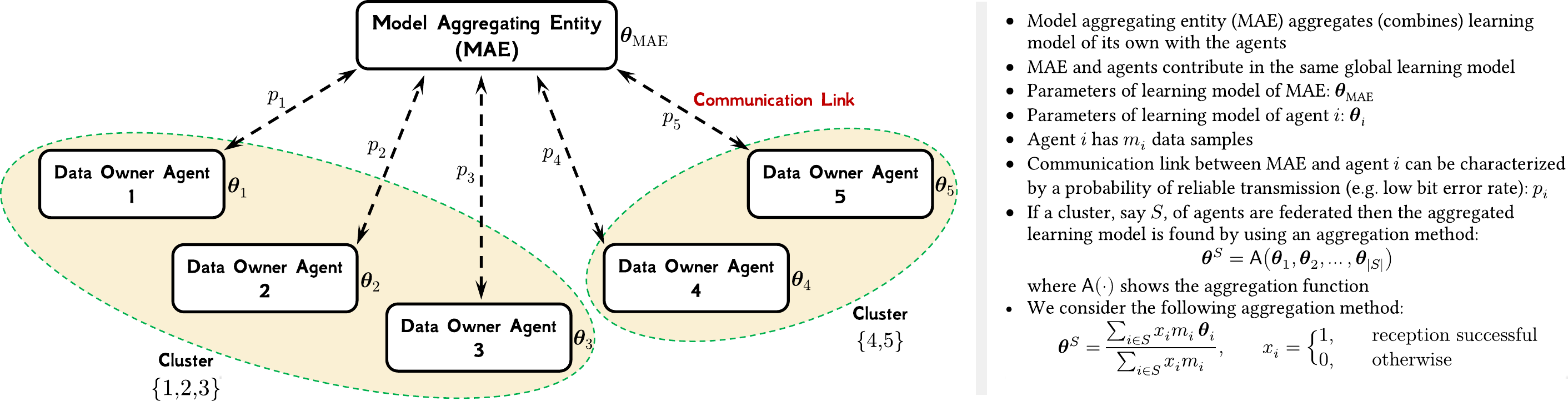}
    \caption{Federated learning framework.}
    \label{fig:model}
\end{figure*}
We consider a set of \textit{agents} denoted $N = \{1,2,\ldots,n\}$ that can participate in the federated learning setting, and a \textit{model aggregating entity} (MAE) which aggregates (combines) \textit{learning model} of its own with the agents. MAE and agents contribute in the same global learning model. The parameters of learning model of MAE and agent $i$ are represented by $\lm_{\text{MAE}} = (\theta_{\text{MAE},1}, \theta_{\text{MAE},2}\ldots,\theta_{\text{MAE},M})$ and $\lm_i = (\theta_{i,1}, \theta_{i,2}, \ldots, \theta_{i,M})$, respectively where $M$ is the number of trainable weights (variables) of the learning model, and every agent $i$ has $m_i$ data samples. We also consider that the communication link between MAE and agent $i$ can be characterized by a probability of reliable transmission (e.g. low bit error rate) denoted  $p_i$. The vector of probabilities of reliable transmission of all agents is shown by $\bm{p} = (p_1,p_2,\ldots,p_n)$. If a \textit{cluster}, say $S\subseteq N$, of agents \textit{agree} to be federated, then the aggregated (combined) learning model is found by using  an aggregation method which is given by
\begin{equation}
 (\lm^S;\bx) = \mathsf{A}(\lm_1,\lm_2,\ldots,\lm_{n_S};\bx), \quad (n_S = |S|)
\end{equation}
where $\mathsf{A}(\cdot)$ shows the aggregation function given $\lm_1,\lm_2,\ldots,\lm_{n_S}$ and $\bx = (x_1,x_2,\ldots,x_n)$ in which $x_i=1$ if MAE receives successfully information of $\lm_i$. This is nothing more than choosing agent $i$ with probability $p_i$. 


\subsection{Incentives of Agents}
We assume that agents may not agree to be in the same cluster depending on their preferences. Thus, we come up with the case where multiple disjoint clusters may occur. In Figure \ref{fig:model}, we illustrate such an example scenario in which two disjoint clusters, i.e. $\{1,2,3\}$ and $\{4,5\}$, create two different aggregated learning models, i.e. $(\lm^{\{1,2,3\}};\bx)$ and $(\lm^{\{4,5\}};\bx)$, respectively.

From the perspective of agents, we assume that MAE assigns a \textit{gain} to all possible clusters. Note that from the MAE point of view, this is the \textit{cost} that must be paid to the cluster. By this way, MAE evaluates the contribution to the aggregated model. However, we shall formulate the problem using the ``gain'' term due to the fact that from the agents point of view, this corresponds to the earnings of the agents. Then, we come up with the question how to design the incentives in order that the agents are willing to participate in federated learning setting taking into account their preferences. 

In this work, we consider the following linear aggregation method:
\begin{equation}\label{eq:aggregatedlossbound}
    (\lm^S;\bx) =\frac{\sum_{i\in S}x_i m_i\lm_i}{\sum_{i\in S}x_i m_i}, \quad  
    x_i = 
    \begin{cases}
        1, & \mbox{reception successful},\\
        0, & \mbox{otherwise}.
    \end{cases}
\end{equation}
which essentially corresponds to the weighted average of the learning models within the cluster. Furthermore, we represent by $\loss(\lm^S;\bx)$ the \textit{loss} of learning model given by parameters $\lm^S$. The expected value of loss function is given by
\begin{gather}
    \E_{\bx}[\loss(\lm^S)] = \sum_{\bx\in\mathcal{X}} \loss(\lm^S;\bx) \prob[{\bx}] \\
    \prob[{\bx}] = \prod_{i\in N} p_i^{x_i}(1-p_i)^{1-x_i}
\end{gather}
where $\mathcal{X}$ with $|\mathcal{X}| = 2^n$ is the set of all possible combinations of $\bx$ vectors. On the other hand, given $\bx$ and cluster $S$, the loss of aggregated model due to $S$ is lower than the loss averaged over the agents in $S$:
\begin{equation}
    \loss(\lm^S;\bx) \leq \frac{\sum_{i\in S}x_i\loss(\lm_i)}{\sum_{i\in S}x_i}
\end{equation}
due to the fact that when the disjoint agents are merged, the amount of data used to train the model increases which results in lower training error. If two disjoint clusters $S$ and $T$, i.e. $S\cap T = \emptyset$, are federated, then we denote the new parameters as $\lm^{S\cup T}$. It is reasonable to assume that the minimum of loss of $\lm^{S\cup T}$ is lower than the minimum of average loss of $\lm^S$ and $\lm^T$:
\begin{equation}
    \loss(\lm^{S\cup T};\bx) \leq \tfrac{\sum_{i\in S}x_i}{\sum_{i\in S\cup T}x_i} \loss(\lm^S;\bx) + \tfrac{\sum_{i\in T}x_i}{\sum_{i\in S\cup T}x_i} \loss(\lm^T;\bx),
\end{equation}

Moreover, we consider that there exists a communication cost, denoted $c$, when MAE receives the learning model's parameters' data; note that this data increases with the size of cluster. On the other hand, MAE earns a monetary gain by utilizing the aggregated model and commits a monetary value which can be paid to the agents. Then, the agents deduct the communication cost $c$ from what they earn from MAE. We represent by $ u $ the gain which assigns a real value for every subset of $ N $, i.e. $ u : 2^N \rightarrow \mathbb{R} $ where $ 2^N $ is the collection of all possible non-empty subsets of $ N $ and empty set $ \emptyset $, and we set $ u(\emptyset) = 0 $. Thus, the gain of cluster $S\in 2^N$ is given by
\begin{equation}\label{eq:gain_of_cluster}
    u(S) = f\left( \tfrac{1}{\E_{\bx}[\loss(\lm^S)]} \right) - c(S).
\end{equation}
where $f:\mathbb{R} \rightarrow \mathbb{R}$ can be a monotonically increasing function and inversely proportional to $\E_{\bx}[\loss(\lm^S)]$ meaning that the less loss the more gain. Note that this is the monetary value that cluster $S$ earns.
Any agent $i$ can join a cluster if guaranteed to be paid at least $u(i) = \pi_i$ which is the \textit{minimal price} given by 
\begin{equation}
    \pi_i = f \left( \tfrac{p_i}{\loss(\lm_i)} \right).
\end{equation}
Note that $p_i$ is the second parameter which has an impact on the price asked by the agent. It corresponds to the fact that as $p_i$ has a poor value, the agent asks lower price to participate in the federation. 

\subsection{Optimal Clustering}
Consider that MAE aims at finding the clustering that results in minimal cost. Optimal clustering problem is defined through the agents $ N $ and 
a \emph{clustering} set $\Pi$ which partitions the agents' set $N$ such that $\bigcup_{S\in \Pi} S = N$. All clusters in $\Pi$ are disjoint clusters, i.e., $S\cap T = \emptyset$ for all $S,T \in \Pi$. Given $\mathcal{P}$, the set of all possible clustering structures, the optimal clustering problem is to find a clustering $ \Pi\in \mathcal{P} $ which minimizes the objective while satisfying the constraints of agents:
\begin{align}\label{eq:optimalclustering}
	&\min_{\Pi\in \mathcal{P}} \sum_{S\in \Pi} u(S) \mbox{ subject to} \notag \\
	&\sum_{i\in S}\pi_i \leq u(S), \quad \forall S \in \Pi,
\end{align}
where the constraints in eq. \eqref{eq:optimalclustering} ensure that the demand of agents are satisfied. The optimization only tells that the agents are guaranteed to be paid their minimal price but not how much more if the gain of cluster allows it. 

\subsection{Clustering under Selfishness}
As agents can behave selfishly, the fundamental question is to find clusters which are stable under selfishness. Let us consider that agent $i$ shall get some monetary gain by joining cluster $S$ as following:
\begin{equation}
    \mbox{gain of agent } i = \pi_i + \phi_i^S
\end{equation}
where $\phi_i^S\in\mathbb{R}$ is the \textit{clustering gain} of agent $i$ by joining cluster $S$, and we set $\phi_i^i = 0$ for all $i\in N$. Such a setting enables to deal with the clustering gains. Thus, the fundamental problem becomes to
\begin{quote}
    \textit{find clustering gains $\phi$ so that the agents agree not to change their cluster}
\end{quote}
Obviously, this is the stable clustering problem under selfishness where the agents strategically decide to which cluster to join; thus, we can define the problem as a \textit{coalition formation game}. We then change the language of problem formulation using game theoretic terms, i.e.
\begin{center}
    \textbf{agents} $\rightarrow$ \textbf{players} \\
    \textbf{cluster} $\rightarrow$ \textbf{coalition} \\  \textbf{clustering} $\rightarrow$ \textbf{partition}
\end{center}
In the sequel, we deal with figuring out family of coalition formation games that ensure stable clusterings.

\section{Hedonic Game}
A hedonic coalition formation game (in short, hedonic game) is given by a pair $\langle N,\succ\rangle$, where $\succ := (\succeq_1,\succeq_2, \ldots ,\succeq_n)$ denotes the \textit{preference profile}, specifying for each \textit{player} $i\in N$ his \textit{preference relation} $\succeq_i$, i.e. a reflexive, complete and transitive binary relation. 

Given $\Pi$, called as \textit{coalition partition}, and $i$, $S_{\Pi}(i)$ denotes the set $S \in \Pi$ such that $i\in S$. Moreover, $ \mathcal{P} $ is the set of all possible coalition partitions over $ N $. 
In its partition form, a coalition formation game is defined on the set $N$ by associating a gain $u(S|\Pi)$ to each subset of any partition $\Pi$ of $N$. The gain of a set is independent of the other coalitions, and therefore, $u(S|\Pi) = u(S)$. The games of this form are more restrictive but present interesting properties to reach a stability. Practically speaking, this assumption means that the gain of a group is independent of the other players outside the group. Hedonic games fall into this category with an additional assumption:
\begin{definition}
	\label{def:hedonic}
	A coalition formation game is  \textit{hedonic} if 
	\begin{itemize}
		\item \emph{the gain of any player depends solely on the members of the coalition to which the player belongs}, and
		\item \emph{the coalitions arise as a result of the preferences of 
			the players over their possible coalitions' set}.
	\end{itemize}
\end{definition}

\subsection{Preference Relation}
The preference relation of a player can be defined over a \emph{preference function}. We consider the case where the preference relation is chosen to be the gain allocated to the player in a coalition. Thus, player $i$ prefers the coalition $S$ to $T$ iff,
\begin{equation}
	\phi_i^S \geq \phi_i^T \Leftrightarrow S \succeq_i T.
\end{equation}

\subsection{The Nash Stability}
The stability concepts for a hedonic game are various. In the literature, a hedonic game could be \textit{individually stable, Nash stable, core stable, strict core stable, Pareto optimal, strong Nash stable, or, strict strong Nash stable}. We refer to 
\cite{bib:Existenceofstabilityinhedoniccoalitionformationgames} for a thorough definition of these different stability concepts. In this paper, we are only interested in the \textbf{Nash stability} because the players do not cooperate to take their decisions jointly. 

\begin{definition}[Nash Stability]
	\label{def:nashstability}
	A partition of players is Nash-stable whenever no player has incentive to unilaterally change its coalition to another coalition in the partition which can be mathematically formulated as follows: partition $\Pi^{\text{NS}}$ is said to be Nash-stable if no player can benefit from moving from his coalition $S_{\Pi^{\text{NS}}}(i)$ to another existing coalition $T\in \Pi^{\text{NS}}$, i.e.:
	\begin{equation}
		S_{\Pi^{\text{NS}}}(i) \succeq_i T \cup i, \quad   \forall T \in \Pi^{\text{NS}} \cup \emptyset; \forall i\in N.
	\end{equation}
	which can be similarly defined over preference function as follows:
	\begin{equation}
		\phi_i^{S_{\Pi^{\text{NS}}}(i)} \geq \phi_i^{T \cup i}, \quad   \forall T \in \Pi^{\text{NS}} \cup \emptyset; \forall i\in N.
	\end{equation}	
\end{definition}
Nash-stable partitions are immune to individual movements even when a player who wants to change does not need permission to join or leave an existing coalition \cite{bogomonlaia}. 
\begin{remark}
Stability concepts being immune to individual deviation are \emph{Nash stability, individual stability, contractual individual stability}. Nash stability is the strongest within above. The notion of \emph{core stability} has been used already in some models where immunity to coalition deviation is required \cite{bib:hajdukova}. 
\end{remark}

\begin{remark}
	In \cite{bib:Barber}, the authors propose some set of axioms which are \textit{non-emptiness, symmetry pareto optimality, self-consistency}; and they analyze the existence of any stability concept that can satisfy these axioms. It is proven that for any game $ |N|>2 $, there does not exist any solution which satisfies these axioms.
\end{remark}
\subsection{Aggregated Learning Model Parameters}
When a stable partition exists, then this means that all the players (agents) are agreed to participate to federation. As a result of this, MAE utilizes the following aggregation of learning model parameters:
\begin{equation}
     \lm^F = w\lm_{\text{MAE}} + (1-w) \frac{\sum_{i\in N}x_i m_i\lm_i}{\sum_{i\in N}x_i m_i}
\end{equation}
where $0 \leq w \leq 1$ is a weighting parameter showing how much MAE favors the aggregated learning model parameters of agents (players), $\lm_{\text{MAE}}$ shows the learning parameters of MAE's local model. In summary, we have the following procedure:  
\begin{oframed}
\noindent\textbf{while} Nash-stable partition $\Pi^{\text{NS}}$ exists
\begin{enumerate}
    \item[1.] Player (agent) $i$ sends information of $\lm_i$, for all $i \in N$
    \item[2.] MAE calculates aggregated learning model parameters $\lm^F$
\end{enumerate}
\textbf{end} 
\end{oframed}
Given $\lm^F$, the expected value of loss function in federation can be calculated as following:
\begin{align}
    \E_{\bx}[\loss(\lm^F)] &= \sum_{\bx\in\mathcal{X}}\loss(\lm^F; \bx) \prob[\bx] \notag \\
    &\geq \loss(\E_{\bx}[\lm^F;\bx]) \qquad (\text{Jensen's inequality})
\end{align}
where 
\begin{equation}
    \E_{\bx}[\lm^F;\bx] = w\lm_{\text{MAE}} + (1-w) \sum_{\bx\in \mathcal{X}} \frac{\sum_{i\in N}x_i m_i\lm_i}{\sum_{i\in N}x_i m_i} \prob[\bx].
\end{equation}
Note that calculating $\E_{\bx}[\loss(\lm^F)]$ may be more difficult than $\E_{\bx}[\lm^F;\bx]$. Therefore, it can be also an option to define the gain of a cluster using $\E_{\bx}[\lm^F;\bx]$ in eq. \eqref{eq:gain_of_cluster}.

\section{The Nash-stable Set}
As the gain $u$ associated with all possible coalitions are known, we are interested in finding a gain distribution to ensure Nash stability. We thus define an \textit{allocation method} $\bm{\phi}\in \mathbb{R}^{\kappa}$ where $\kappa = n2^{n-1}$ as following:
\begin{equation}\label{eq:utilityallocationmethod}
	\bm{\phi}=\{\phi_i^S : \forall i\in S, \forall S \in 2^N\}
\end{equation}
which directly sets up a preference profile. The set of all possible  allocation methods is denoted by $\mathcal{F}\subset \mathbb{R}^{\kappa}$. We define the \textit{mapping} $\mathsf{M}$, which for any allocation method $\bm{\phi}$, it finds corresponding all possible Nash-stable partitions, i.e. $\mathsf{M}(\bm{\phi}) \subset \mathcal{P}$. 

We define the Nash-stable set which includes all those allocation methods  that build the following set:
\begin{multline}\label{eq:theNashStabelSet}
	\NashStableSet = \left\{ \bm{\phi}\in \mathbb{R}^{\kappa} : \exists \Pi \in \mathsf{M}(\bm{\phi}) | S_{\Pi}(i) \succeq_i T\cup i, \right.  \\ \left. \forall T\in \Pi\cup \emptyset; \forall i\in N \right\}.
\end{multline}
Essentially, the Nash-stable set includes 
\begin{quote}
\textit{the family of hedonic games, each one having a different preference profile that derives from a different allocation method. Thus, before finding a Nash-stable partition, we need to find the hedonic game (i.e., an allocation method) for which a Nash-stable partition exists.}
\end{quote}

Let us define the set of constraints stemming from the preference function in order to check if the Nash-stable set is non-empty. Due to the gain bound, for any allocation method $\bm{\phi}$, we have 
\[
\sum_{i\in S}(\pi_i+\phi_i^S) \leq u(S), \quad \forall S\in 2^N
\]
called as \textit{budged balanced} gain allocation which further can be given by
\begin{equation}\label{eq:budgetbalancedness}
    \sum_{i\in S} f\left( \tfrac{p_i}{\loss(\lm_i)} \right) + \sum_{i\in S} \phi_i^S \leq f\left( \tfrac{1}{\E_{\bx}[\loss(\lm^S)]}\right) - c(S), \quad \forall S\in 2^N.
\end{equation}
For simplicity, let us define \textit{marginal gain} as following:
\begin{equation}
\Delta_{\lm}(S) = 
    \begin{cases}
    f\left( \tfrac{1}{\E_{\bx}[\loss(\lm^S)]} \right) - c(S) - \sum_{i\in S}f\left( \tfrac{p_i}{\loss(\lm_i)} \right),  &\forall S\in 2^N\setminus i,\\
    0, & \forall i\in N.
    \end{cases}
\end{equation}
which results in the following constraints:
\begin{equation}\label{eq:BudgetBalancednessConstraints}
	\mathscr{C}_{\lm}^1(\bm{\phi}):=\left\{\sum_{i\in S}\phi_i^S \leq \Delta_{\lm}(S), \forall S \in 2^N \right\}.
\end{equation} 
which are the constraints that stem from budged balancedness. On the other hand, for any $\bm{\phi}$, the constraints that ensure the Nash stability are given by
\begin{equation}\label{eq:NashStabilityConstraints}
	\mathscr{C}_{\lm}^2(\bm{\phi}) := \left\{\exists \Pi \in \mathsf{M}(\bm{\phi})  \left|  \phi_i^{S_{\Pi}(i)}\geq \phi_i^{T\cup i}, \forall T\in \Pi\cup \emptyset; \forall i \in N\right.\right\},
\end{equation}
Based on these two constraints represented by $\mathscr{C}_{\lm}^1(\bm{\phi})$ and $\mathscr{C}_{\lm}^2(\bm{\phi})$, we can define the Nash-stable set. 
\begin{equation}\label{def:NashStableCore}
	\NashStableSet(\lm) = \left\{ \bm{\phi}\in \mathbb{R}^{\kappa} :   \mathscr{C}_{\lm}^1(\bm{\phi}) \mbox{ and } \mathscr{C}_{\lm}^2(\bm{\phi})  \right\},
\end{equation}
Then, the non-emptiness of the Nash-stable set is crucial. The theorem below states the necessary conditions about the non-emptiness of the Nash-stable set:
\begin{oframed}
\begin{theorem}\label{thm:NonEmptinessOfNashstableSet}
	The Nash-stable set can be non-empty.
\end{theorem}
\begin{proof}
We can check if the Nash-stable set is non-empty by solving the following optimization problem:
\begin{align*}
    &\max_{\bm{\phi}} \sum_{S\in 2^N} \sum_{i\in S}  \phi_i^S \mbox{ subject to } \mathscr{C}_{\lm}^1(\bm{\phi}) \mbox{ and } \mathscr{C}_{\lm}^2(\bm{\phi}).
\end{align*}
If there exists any feasible solution of this problem, then we conclude that there is at least one allocation method which provides a Nash-stable partition.
\end{proof}
\end{oframed}
However, searching in an exhaustive manner over the whole partitions is NP-hard as the number of partitions grows according to the Bell number. Typically, with only $ 10 $ players, the number of partitions is as large as $115,975$. 

\subsection{Superadditive Gain}
If the gain function $u$ is superadditive, then it is trivial to check that the marginal gain is also superadditive:
$\Delta_{\lm}(S\cup T) \geq \Delta_{\lm}(S) + \Delta_{\lm}(T)$, for all possible $S$ and $T$ such that $S\cap T = \emptyset$. Due to eq. \eqref{eq:budgetbalancedness}, we have
\begin{gather*}
    \sum_{i\in S\cup T}\phi_i^{S\cup T} = \sum_{i\in S}\phi_i^{S\cup T} + \sum_{i\in T}\phi_i^{S\cup T} \leq \Delta_{\lm}(S \cup T) \\
    \sum_{i\in S}\phi_i^{S} + \sum_{i\in T}\phi_i^{T} \leq \Delta_{\lm}(S) + \Delta_{\lm}(T) \\
    \Rightarrow \sum_{i\in S}\phi_i^{S\cup T} + \sum_{i\in T}\phi_i^{S\cup T} \geq \sum_{i\in S}\phi_i^{S} + \sum_{i\in T}\phi_i^{T}.
\end{gather*}
This result means that any player is better off in a larger coalition which ultimately all players have the most gain in the grand coalition. This is obvious from eq. \eqref{eq:NashStabilityConstraints} where for every player $i\in N$, $\phi_i^N \geq \phi_i^S$ for all $S\in 2^N$.

\subsection{Additively Separable and Symmetric Gain}\label{sec:AdditivelySeparableandSymmetric}
Preferences of a player are \emph{additively separable} whenever the preference can be  stated with a function characterizing how a player prefers another player in each coalition. This means that the player's preference for a coalition is based on individual preferences. This can be formalized as follows:
\begin{definition}
	The preferences of a player are said to be \textbf{additively
		separable} if there exists a function $v_i:N \rightarrow \mathbb{R}$ such that
	\begin{equation}\label{eq:additivelyseparable}
		\sum_{j\in S} v_i(j) \geq \sum_{j\in T} v_i(j) \Leftrightarrow S \succeq_i T, \quad \forall S,T\in  2^N.
	\end{equation}
\end{definition}
$v_i(i)$ is normalized and set to $v_i(i)=0$. A profile of additively separable preferences satisfies \emph{symmetry} if $v_i(j) = v_j(i) = v(i,j)$, for all $i,j\in N$. The meaning of $v(i,j)$ is the \textit{mutual gain} of player $i$ and $j$ when they are in the same coalition. Let $\mathcal{V}(S)$ be the all possible bipartite coalitions which can occur in coalition $S$ such that:
\[
\mathcal{V}(S):= \{(i,j)\in S: j>i\}, \quad |S| \geq 2.
\]
We then define ${\bf v}\in \mathbb{R}^{|\mathcal{V}(N)|}$ which shall serve as an allocation method to generate additively separable and symmetric preferences:
\[
    {\bf v} = \left\{ v(i,j) : \forall (i,j)\in \mathcal{V}(N) \right\}
\]
and mapping $\mathsf{M}$ shall find all possible Nash-stable partitions,, i.e. $\mathsf{M}({\bf v}) \subset \mathcal{P}$. The constraints that define the Nash-stable set are then defined over ${\bf v}$:
\[
\mathscr{C}_{\lm}^1(\bm{\phi}) \rightarrow \mathscr{C}_{\lm}^1({\bf v}) \mbox{ and } \mathscr{C}_{\lm}^2(\bm{\phi}) \rightarrow \mathscr{C}_{\lm}^2({\bf v})
\]
Further, note that the gain that player $i$ has in coalition $S$ is given by
\begin{gather}
    \pi_i + \phi_i^S = f\left( \tfrac{1}{\loss(\lm_i)} \right) + \sum_{j\in S} v(i,j) \notag \\
    \Rightarrow \phi_i^S = \sum_{j\in S} v(i,j)
\end{gather}
On the other hand, due to the symmetry property of mutual gain, we have the following:
\[
\sum_{i,j\in S} v(i,j) = 2\sum_{(i,j)\in \mathcal{V}(S)} v(i,j).
\] 
For example, if $S=(1,2,3)$, then $\sum_{i,j\in S} v(i,j) = 2[v(1,2) + v(1,3) + v(2,3)]$. 
\begin{oframed}
\begin{theorem}
    Additively separable and symmetric preferences always admit a Nash-stable partition. Therefore, constraints in $\mathscr{C}^2_{\lm}({\bf v})$ are always satisfied \cite{bib:hajdukova}.
\end{theorem}
\end{oframed}
Based on this theorem, we only need to satisfy the constraints given by $\mathscr{C}_{\lm}^1({\bf v})$. Thus, we define the Nash-stable set which generates additively separable and symmetric preferences $\NashStableSet^{\text{A}}(\lm)\subset \NashStableSet(\lm)$ as following:
\begin{equation}\label{eq:NashStableSetAdditivelySeparableAndSymmetric}
    \NashStableSet^{\text{A}}(\lm) = \Bigg\{ {\bf v}\in \mathbb{R}^{|\mathcal{V}(N)|} : \underbrace{
	\sum_{(i,j)\in \mathcal{V}(S)} v(i,j) \leq \tfrac{\Delta_{\lm}(S)}{2} , \forall S \in 2^N}_{\mathscr{C}^1_{\lm}({\bf v})}
	\Bigg\}
\end{equation}
Finding the values of $ v(i,j) $ in eq. (\ref{eq:NashStableSetAdditivelySeparableAndSymmetric}) satisfying $\mathscr{C}_{\lm}^1({\bf v})$ conditions can be done straightforward. However, we propose to formulate as an optimization problem for finding the values of $ v(i,j)$. A feasible solution of the following linear program guarantees the non-emptiness of $\NashStableSet^{\text{A}}(\lm)$:
\begin{align}\label{eq:OptimalAdditivelySeparableAndSymmetric}
    &\max_{{\bf v}} \sum_{(i,j)\in \mathcal{V}(N)} v(i,j) \mbox{ subject to } \notag \\
    &\sum_{(i,j)\in \mathcal{V}(S)} v(i,j) \leq \tfrac{\Delta_{\lm}(S)}{2}, \quad \forall S \in 2^N, 
\end{align}
where note that any feasible solution ${\bf v}^*$ is upper bounded by $\sum_{(i,j)\in \mathcal{V}(N)} v^*(i,j) \leq {\Delta_{\lm}(N)}/{2}$. Furthermore, the coalition partition that stems from ${\bf v}^*$ is given by $\Pi^{\text{NS}} \in \mathsf{M}({\bf v}^*)$ which is Nash-stable.


\section{Decentralized Clustering}
In this section, we study finding a Nash-stable partition in a decentralized setting which corresponds to finding stable decentralized clustering. We, in fact, model the problem of finding a Nash-stable partition as a non-cooperative game. 

A hedonic coalition formation game is equivalent to a non-cooperative game. Denote as $\Sigma$ the set of \textit{strategies}. We assume that the number of strategies is equal to the number of players, i.e. $ |\Sigma|=n $. This is sufficient to represent all possible choices. Indeed, the players that select the same strategy are interpreted as a coalition. For example, if every player chooses different strategies, then this corresponds to the coalition partition comprised of singletons.

Consider the \textit{best-reply dynamics} where in a particular step, only one player chooses its best strategy. A \textit{strategy tuple} is represented as $\bm{\sigma} = \{\sigma_1, \sigma_2, \ldots, \sigma_n\}$, where $\sigma_i\in \Sigma$ is the strategy of player $i$. In every step, only one dimension is changed in $\bm{\sigma}$. We further define 
\begin{align}
    & S_{\bm{\sigma}}(i) = \{ j\in N: \sigma_i = \sigma_j \} \\
    & \Pi(\bm{\sigma}) = \{ S_{\bm{\sigma}}(i), \forall i \in N \}
\end{align}
the set of players that share the same strategy with player $i$ and partition of players with respect to strategy tuple $\bm{\sigma}$. Thus, note that $\cup_{i\in N} S_{\bm{\sigma}}(i) = N$ for each step. The gain of player $i$ in case of strategy tuple $\bm{\sigma}$ is represented by $\phi_i(\bm{\sigma})$ which verifies the following relation:
\begin{equation}
	\phi_i(\bm{\sigma}) \geq \phi_i(\bm{\sigma}') \Leftrightarrow S_{\bm{\sigma}}(i) \succeq_i S_{\bm{\sigma}'}(i), 
\end{equation}
Any sequence of strategy tuple in which each strategy tuple differs from the preceding one in only one coordinate is called a \textit{path}, and a unique deviator in each step strictly increases the gain he receives is an \textit{improvement path}. Obviously, any \textit{maximal improvement path} which is an improvement path that can not be extended is terminated by stability.

\subsection{Equilibrium Analysis}
The Nash equilibrium is defined as following:
\begin{equation}
	\sigma_i^{\text{NE}} \in \arg \max_{\sigma_i\in \Sigma} \phi_i(\sigma_i,\sigma_{-i}), \quad \forall i\in N.
\end{equation}
essentially corresponding to a Nash-stable partition in the original hedonic game, which is given by
\begin{align}
	& S_{\bm{\sigma}^{\text{NE}}}(i) = \{j\in N: \sigma^{\text{NE}}_i = \sigma_j\}, \quad \forall i \in N \notag \\ 
	& \Pi^{\text{NE}} = \{S_{\bm{\sigma}^{\text{NE}}}(i), \forall i\in N\}.
\end{align}

In the sequel, we prove that the additively separable and symmetric gains result in a \textit{potential game} where all the players have incentive to change their strategy according to a single global function called as \textit{potential function}.
\begin{oframed}
\begin{theorem}\label{thm:AdditivelySeparableAndSymmetricUtilitiesResultsInAPotentialGame}
	Any additively separable and symmetric gain results in a potential game with potential function:
	\begin{equation}\label{eq:additivelyseparableandsymmetricpotential}
		P_{{\bf v}}(\bm{\sigma}) = \sum_{S\in \Pi(\bm{\sigma})} \sum_{(i,j)\in \mathcal{V}(S)} v(i,j).
	\end{equation}
\end{theorem}
\begin{proof}
	A non-cooperative game is a potential game whenever there exists a function $P_{{\bf v}}$ such that:
	\begin{equation*}
	P_{{\bf v}}(\sigma_i,\sigma_{-i}) - P_{{\bf v}}(\sigma'_i,\sigma_{-i}) = \phi(\sigma_i,\sigma_{-i}) - \phi(\sigma'_i,\sigma_{-i})
	\end{equation*}
	where $(\sigma_i,\sigma_{-i}) = \bm{\sigma}$ and $\sigma_{-i}$ shows the strategies of the players other than $i$. This means that when player $i$ switches from strategy $\sigma_i$ to $\sigma'_i$ the difference of its gain can be given by the difference of a function $P$. We choose the following potential function:
	\begin{equation}\label{eq:additivelyseparableandsymmetricpotential2}
	P_{{\bf v}}(\bm{\sigma}) = \sum_{S\in \Pi(\sigma)} \sum_{(i,j)\in \mathcal{V}(S)} v(i,j)
	\end{equation}
Let us denote as $i\in S$ and $i\not\in S'$ the coalitions when player $i$ switches from strategy $\sigma_i$ to $\sigma'_i$, respectively. Potential function is given as following before and after switching
	\begin{gather*}
		P_{{\bf v}}(\sigma_i,\sigma_{-i}) =  \sum_{(i,j)\in \mathcal{V}(S)} v(i,j) + \sum_{(k,j)\in \mathcal{V}(S')} v(k,j) \\ + \sum_{T\in \Pi(\bm{\sigma})\setminus \{S,S'\}} \sum_{(k,j)\in \mathcal{V}(T)} v(k,j) 
	\end{gather*}
	\begin{gather*}
		P_{{\bf v}}(\sigma'_i,\sigma_{-i}) =  \sum_{(k,j)\in \mathcal{V}(S\setminus i)} v(k,j) + \sum_{(k,j)\in \mathcal{V}(S'\cup i)} v(k,j) \\ + \sum_{T\in \Pi(\sigma)\setminus \{S,S'\}} \sum_{(k,j)\in \mathcal{V}(T)} v(k,j)
	\end{gather*}
	where note that we have $S \rightarrow S\setminus i$ and $S' \rightarrow S'\cup i$ after switching. Thus, we have
	\begin{gather*}
	P_{{\bf v}}(\sigma_i,\sigma_{-i}) - 	P_{{\bf v}}(\sigma'_i,\sigma_{-i}) = \\  \sum_{(i,j)\in \mathcal{V}(S)} v(i,j) + \sum_{(k,j)\in \mathcal{V}(S')} v(k,j) \\ - \sum_{(k,j)\in \mathcal{V}(S\setminus i)} v(k,j) - \sum_{(k,j)\in \mathcal{V}(S'\cup i)} v(k,j) = \\
	\sum_{j\in S} v(i,j) -\sum_{j\in S'\cup i} v(i,j)
	\end{gather*}
	On the other hand, the gain shift before and after strategy switch is given by
	\begin{gather*}
	\phi(\sigma_i,\sigma_{-i}) - 	\phi(\sigma'_i,\sigma_{-i}) = 
	\sum_{j\in S} v(i,j) -\sum_{j\in S'\cup i} v(i,j)
	\end{gather*}
	which concludes the proof that $P_{{\bf v}}(\sigma_i,\sigma_{-i}) - P_{{\bf v}}(\sigma'_i,\sigma_{-i}) = \phi(\sigma_i,\sigma_{-i}) - \phi(\sigma'_i,\sigma_{-i})$. 	
\end{proof}
\end{oframed}
In a potential game, a Nash equilibrium shall result in an optimum in potential $P_{{\bf v}}$. Therefore, $\bm{\sigma}^{*} \in \arg\max_{\bm{\sigma}} P_{{\bf v}}(\bm{\sigma})$ corresponds to a coalition partition $\Pi(\sigma^*) \in \mathsf{M}({\bf v})$ which is Nash-stable.

\section{Conclusions}
We analyzed stable clustering problem in federated learning setting. Clusters are made up of the agents contributing to federated learning. We considered that every agent is better off when switching from one cluster to another one. We modeled the decisions of agents in the framework hedonic games which is a widely used cooperative game model for this type of problems. A fundamental question in hedonic games is to analyze the conditions how stable coalition partitions can occur. We studied the existence of stable coalition partitions by introducing the Nash-stable set, and analyzed the existence of  decentralized coalition partitions. 

As future work, it may be interesting to do stability analysis for different stability notions, to study other types of preference profiles ensuring Nash stability. On the other hand, it is essential to do experiments with real data and specific learning models as well as realistic gain and communication cost functions.



\bibliographystyle{ACM-Reference-Format} 
\bibliography{main}


\end{document}